\documentclass[english]{article}
\usepackage[T1]{fontenc}
\usepackage[latin9]{inputenc}
\usepackage{geometry}
\usepackage{babel}
\usepackage{float}
\usepackage{amsmath}
\usepackage{amsthm}
\usepackage{amssymb}
\usepackage{stackrel}
\usepackage[unicode=true]
 {hyperref}

\makeatletter
\newcommand{\lyxaddress}[1]{
	\par {\raggedright #1
	\vspace{1.4em}
	\noindent\par}
}
\theoremstyle{plain}
\newtheorem{thm}{\protect\theoremname}
\theoremstyle{definition}
\newtheorem{defn}[thm]{\protect\definitionname}
\theoremstyle{definition}
\newtheorem{example}[thm]{\protect\examplename}
\theoremstyle{plain}
\newtheorem{cor}[thm]{\protect\corollaryname}
\theoremstyle{remark}
\newtheorem{claim}[thm]{\protect\claimname}
\theoremstyle{remark}
\newtheorem{rem}[thm]{\protect\remarkname}

\@ifundefined{date}{}{\date{}}
\makeatother

\providecommand{\claimname}{Claim}
\providecommand{\corollaryname}{Corollary}
\providecommand{\definitionname}{Definition}
\providecommand{\examplename}{Example}
\providecommand{\remarkname}{Remark}
\providecommand{\theoremname}{Theorem}

\begin{document}
\title{\textbf{How to Find New Characteristic-Dependent Linear Rank Inequalities
using Binary Matrices as a Guide}}
\author{Victor Peña\thanks{e-mail: \protect\href{mailto:vbpenam\%40unal.edu.co}{vbpenam@unal.edu.co}}
and Humberto Sarria\thanks{e-mail: \protect\href{mailto:hsarriaz\%40unal.edu.co}{hsarriaz@unal.edu.co}}}
\date{Departamento de Matemáticas, Universidad Nacional de Colombia, Bogotá,
Colombia}

\maketitle
In Linear Algebra over finite fields, a characteristic-dependent linear
rank inequality is a linear inequality that holds by ranks of subspaces
of a vector space over a finite field of determined characteristic,
and does not in general hold over other characteristics. In this paper,
we show a method to produce these inequalities using binary matrices
with suitable ranks over different fields. In particular, for each
$n\geq7$, we produce $2\left\lfloor \frac{n-1}{2}\right\rfloor -4$
characteristic-dependent linear rank inequalities over $n$ variables.
Many of the inequalities obtained are new but some of them imply the
inequalities presented in \cite{4,14}.

\lyxaddress{\textbf{Keywords: }Linear rank inequality, complementary vector space,
binary matrix.\textbf{\footnotesize{}\vspace{-0.2cm}
}}

\lyxaddress{\textbf{\footnotesize{}Subject Classification: 68P30\vspace{-0.5cm}
}}

\section{Introduction}

A \emph{linear rank inequality} is a linear inequality that is always
satisfied by ranks (dimensions) of subspaces of a vector space over
any field. Information inequalities are a sub-class of linear rank
inequalities \cite{15}. The Ingleton inequality is an example of
a linear rank inequality which is not information inequality \cite{13}.
Other inequalities have been presented in \cite{8,13}. A \emph{characteristic-dependent
linear rank inequality} is as a linear rank inequality but this is
always satisfied by vector spaces over fields of certain characteristic
and does not in general hold over other characteristics. In Information
Theory, especially in linear network coding, all these inequalities
are useful to calculate the linear capacity of communication networks
\cite{9}. It is remarkable that the linear capacity of a network
depends on the characteristic of the scalar field associated to the
vector space of the network codes, as an example, the Fano network
\cite{6,9}. Therefore, when we study linear capacities over specific
fields, characteristic-dependent linear rank inequalities are more
useful than usual linear rank inequalities.

Characteristic-dependent linear rank inequalities have been presented
by Blasiak, Kleinberg and Lubetzky \cite{4}, Dougherty, Freiling
and Zeger \cite{9}, and E. Freiling \cite{11}. The technique used
by Dougherty et al. to produce these inequalities used as a guide
the flow of some matroidal network to obtain restriction over linear
solubility of these and it is different from the technique used by
Blasiak et al. which is based on the dependency relations of the Fano
and non-Fano Matroids. In \cite{14}, we show some inequalities using
the ideas of Blasiak and present some applications to network coding
that improve some existing results in \cite{4,11}.

\textbf{Organization of the work and contributions. }We show a general
method to produce characte-ristic-dependent linear rank inequalities
using as a guide binary matrices with suitable rank over different
fields. We try to find as many inequalities as the method can produce:
For each $n\geq7$, we explicitly produce $2\left\lfloor \frac{n-1}{2}\right\rfloor -4$
characteristic-dependent linear rank inequalities in $n$ variables
of which half are true over characteristics in sets of primes of the
form $\left\{ p:p\mid t\right\} $ and the other half are true over
characteristics in sets of primes of the form $\left\{ p:p\nmid t\right\} $,
where $2\leq t\leq\left\lfloor \frac{n-1}{2}\right\rfloor -1$, but
we note that more inequalities can be produced. Also, for the first
class of inequalities, we prove that all are independent of each other
and they can not be recovered from any of our inequalities in a greater
number of variables. We remark that to date such number of inequalities
of this type in $n$ variables were not known. In addition, the inequalities
presented in \cite{14} can be recovered when $n$ is of the form
$2m+3$ and $t$ is equal to $m$.

\section{Entropy in Linear Algebra}

Let $A$, $A_{1}$, $\ldots$, $A_{n}$, $B$ be vector subspaces
of a finite dimensional vector space $V$ over a finite field $\mathbb{F}$.
Let $\underset{i\in I}{\sum}A_{i}$ denote the span or sum of $A_{i}$.
The sum $A+B$ is a direct sum if and only if $A\cap B=O$, the notation
for such a sum is $A\oplus B$. Subspaces $A_{1}$, ..., $A_{n}$
are called \emph{mutually complementary} subspaces in $V$ if every
vector of $V$ has an unique representation as a sum of elements of
$A_{1}$, ..., $A_{n}$. Equivalently, they are mutually complementary
subspaces in $V$ if and only if $V=A_{1}\oplus\cdots\oplus A_{n}$.
In this case, $\pi_{S}$ denotes the canonical projection function
$V\twoheadrightarrow\underset{i\in S}{\bigoplus}A_{i}$. $\left\{ e_{i}\right\} $
is the canonical bases in $V$ and $e_{S}$ is the vector whose inputs
are $1$ in the components in $S$ and $0$ in another case.

There is a correspondence between linear rank inequalities and information
inequalities associated to certain class of random variables induced
by vector spaces \cite[Theorem 2]{15}, we explain that: Let $f$
be chosen uniformly at random from the set of linear functions from
$V$ to $\mathbb{F}$, for $A_{1}$, $\ldots$, $A_{n}$ define the
random variables $X_{1}=f\mid_{A_{1}}$, $\ldots$, $X_{n}=f\mid_{A_{n}}$,
then
\[
\mathrm{H}\left(X_{i}:i\in I\right)=\log\left|\mathbb{F}\right|\dim\left(\underset{i\in I}{\sum}A_{i}\right),\text{ \ensuremath{I\subseteq\left[n\right]:=\left\{ 1,\ldots,n\right\} }}.
\]
The difference between entropy and dimension is a fixed positive multiple
scalar. Therefore, any inequality satisfied by entropies, it is an
inequality satisfied by dimensions of vector spaces; for simplicity,
we identify these parameters, i. e. the entropy of $A_{1}$, $\ldots$,
$A_{n}$ is
\[
\mathrm{H}\left(A_{i}:i\in I\right)\overset{\text{def.}}{=}\dim\left(\underset{i\in I}{\sum}A_{i}\right).
\]
So, we can think $A_{1}$, $\ldots$, $A_{n}$ as a tuple of random
variables induced in described form, such random variables are called
linear random variables over $\mathbb{F}$.

The mutual information of $A$ and $B$ is given by $\text{I}\left(A;B\right)=\dim\left(A\cap B\right).$
If $B$ is a subspace of a subspace $A$, then we denote the \emph{codimension}
of $B$ in $A$ by $\text{codim}_{A}\left(B\right):=\text{H}\left(A\right)-\text{H}\left(B\right).$
We have that $\text{H}\left(A\mid B\right)=\text{codim}_{A}\left(A\cap B\right)$.
In a similar way conditional mutual information is expressed.

We formally define the inequalities that concern this paper:
\begin{defn}
Let $m$ be a positive integer, let $P$ be a set of primes, and let
$S_{1}$, $\ldots$, $S_{k}$ be subsets of $\left\{ 1,\ldots,m\right\} $.
Let $\alpha_{i}\in\mathbb{R}$ for $1\leq i\leq k$. A linear inequality
of the form 
\[
\stackrel[i=1]{k}{\sum}\alpha_{i}\mathrm{H}\left(A_{j}:j\in S_{i}\right)\geq0
\]

- is called a \emph{characteristic-dependent linear rank inequality}
if it holds for all jointly distributed linear random variables $A_{1}$,
$\ldots$, $A_{m}$ over finite fields with characteristic in $P$.

- is called a \emph{linear rank inequality} if it is a characteristic-dependent
linear rank inequality with $P$ is equal to the collection of all
prime numbers.

- is called an \emph{information inequality} if the inequality holds
for all jointly distributed random variables.
\end{defn}

The following inequality is the first linear rank inequality which
is not information inequality.
\begin{example}
(Ingleton's inequality \cite{12}) For any $A_{1}$, $A_{2}$, $A_{3}$
and $A_{4}$ vector subspaces of a finite dimensional vector space,
$\mathrm{I}\left(A_{1};A_{2}\right)\leq\mathrm{I}\left(A_{1};A_{2}\mid A_{3}\right)+\mathrm{I}\left(A_{1};A_{2}\mid A_{4}\right)+\mathrm{I}\left(A_{3};A_{4}\right).$
\end{example}

We are interested in finding interesting characteristic-dependent
linear rank inequalities i.e. where $P$ is a proper subset of primes.

\subsection{Producing inequalities: How to find and use a suitable binary matrix}

The following theorem is the principal theorem of this paper and shows
a method to produce pairs of characteristic-dependent linear rank
inequalities from suitable binary matrices. The demonstration is presented
in subsection \ref{subsec:secci=0000F3n de prueba del teorema principal}.
We use this notation: $\left[e_{n},e_{m}\right]=\left\{ e_{i}:n\leq i\leq m\right\} $,
$\left[e_{n},e_{m}\right)=\left\{ e_{i}:n\leq i<m\right\} $ and $\left[e_{n}\right]:=\left[e_{1},e_{n}\right]=\left\{ e_{i}\right\} $;
for a binary matrix $B$, we denote $B=\left(B^{i}\right)=\left(e_{S_{i}}\right)$,
with $S_{i}=\left\{ j:B_{\left(j,i\right)}=1\right\} $.
\begin{thm}
\label{thm: principal teorema o metodo que usa matrices binarias}Let
$B=\left(B^{i}\right)=\left(e_{S_{i}}\right)$ be a $n\times m$ binary
matrix over $\mathbb{F}$, $m\leq n$ and $t\geq2$ integer. We suppose
that $\text{rank}B=m$ if $\text{char}\mathbb{F}$ does not divide
$t$, and $\text{rank}B=m-1$ in other cases. Let $A_{e_{i}}$, $e_{i}\in\left[e_{n}\right]$,
$B_{e_{S_{i}}}$, $e_{S_{i}}\in\mathcal{B}'$ and $C$ be vector subspaces
of a finite dimensional vector space $V$ over $\mathbb{F}.$ Then

(i) The following inequality is a characteristic-dependent linear
rank inequality over fields whose characteristic divides $t$,
\[
\text{H}\left(A_{e_{j}},B_{e_{S_{i}}},C:e_{S_{i}}\in\mathcal{B}',e_{j}\in\mathcal{B}'',C\in\mathcal{B}'''\right)+\left(\left|\mathcal{B}''\right|\left|\mathcal{B}'\right|+\left|\mathcal{B}''\right|\right)\text{H}\left(C\right)\leq\left(m-1\right)\text{I}\left(A_{\left[e_{n}\right]};C\right)
\]
\[
+\underset{e_{S_{k}}\in\mathcal{B}'}{\sum}\left[\text{H}\left(E_{e_{S_{k}}}\mid A_{e_{i}}:i\in S_{k}\right)+\text{H}\left(E_{e_{S_{k}}}\mid A_{e_{i}},C:i\notin S_{k}\right)\right]+\left(\left|\mathcal{B}'\right|+1\right)\underset{e_{i}\in\mathcal{B}''}{\sum}\text{H}\left(A_{e_{i}}\right)
\]
\[
+\left(\left|\mathcal{B}''\right|\left|\mathcal{B}'\right|+\left|\mathcal{B}'''\right|+\left|\mathcal{B}''\right|+\left|\mathcal{B}'\right|\right)\left[\text{H}\left(C\mid A_{\left[e_{n}\right]}\right)+\underset{e_{i}\in\left[e_{n}\right]}{\sum}\text{I}\left(A_{\left[e_{n}\right]-e_{i}};C\right)\right]
\]
\[
+\underset{e_{S_{k}}\in\mathcal{B}'}{\sum}\left[\nabla\left(A_{e_{i}}:i\in S_{k},e_{i}\notin\mathcal{B}''\right)+\nabla\left(A_{e_{i}}:i\notin S_{k},e_{i}\notin\mathcal{B}''\right)\right],
\]

(ii) The following inequality is a characteristic-dependent linear
rank inequality over fields whose characteristic does not divide $t$,
\[
\text{H}\left(C\right)\leq\frac{1}{m}\text{H}\left(A_{e_{j}},B_{e_{S_{i}}},C:e_{S_{i}}\in\mathcal{B}',e_{j}\in\mathcal{B}'',C\in\mathcal{B}'''\right)+\text{H}\left(C\mid A_{\left[e_{n}\right]}\right)+\underset{e_{i}\in\left[e_{n}\right]}{\sum}\text{I}\left(A_{\left[e_{n}\right]-e_{i}};C\right)
\]
\[
+\underset{e_{S_{k}}\in\mathcal{B}'}{\sum}\left[\text{H}\left(C\mid A_{e_{i}},B_{S_{k}}:i\notin S_{k}\right)+\text{H}\left(B_{e_{S_{k}}}\mid A_{e_{i}}:i\in S_{k}\right)+\nabla\left(A_{e_{i}}:i\notin S_{k}\right)+\nabla\left(A_{e_{i}}:i\in S_{k}\right)\right],
\]
where $\mathcal{B}'=\left\{ e_{S_{i}}:1<\left|S_{i}\right|<n\right\} $;
$\mathcal{B}''=\left\{ e_{S_{i}}:\left|S_{i}\right|=1\right\} $;
$\mathcal{B}'''=\left\{ C\right\} $ if there exists $e_{S_{i}}$
in $B$ such that $\left|S_{i}\right|=n$, and $\mathcal{B}'''$ is
empty in other case; and $\nabla$ is a finite sum of entropies given
by
\[
\nabla\left(A_{e_{i}}:e_{i}\in T\subseteq\left[n\right]\right)\overset{def.}{=}\text{I}\left(A_{\left[e_{1},e_{k_{1}}\right)};A_{\left[e_{k_{1}},e_{k_{2}}\right]}\right)+\cdots+\text{I}\left(A_{\left[e_{1},e_{k_{1}-1}\right)};A_{\left[e_{k_{1-1}},e_{k_{l}}\right]}\right)
\]
where $k_{1}\leq k_{2}\leq\cdots\leq k_{l}$ give a partition in intervals,
with maximum length, of $T$.
\end{thm}

The first inequality does not hold in general over vector spaces whose
characteristic does not divide $t$ and the second inequality does
not hold in general over vector spaces whose characteristic divides
$t$. A counter example would be in $V=\text{GF}\left(p\right)^{n}$,
take the vector spaces $A_{e_{i}}=\left\langle e_{i}\right\rangle $,
$e_{i}\in\left[e_{n}\right]$, $B_{e_{S_{j}}}=\left\langle e_{S_{j}}\right\rangle $,
$e_{S_{j}}\in\mathcal{B}'$, and $C=\left\langle \sum e_{i}\right\rangle $
Then, when $p$ does not divide $t$, first inequality does not hold;
and when $p$ divides $n$, second inequality does not hold.
\begin{cor}
\label{cor:corolario sobre cota de la dimension para que las desigualdades sean validas ind de caract}If
the dimension of vector space $V$ is at most $n-1$, then inequalities
implicated in Theorem \ref{thm: principal teorema o metodo que usa matrices binarias}
are true over any field.
\end{cor}

\begin{cor}
\label{cor:corolario sobre eliminai=0000F3n de variable para que la desi sea rango lineal}If
some vector space in Theorem \ref{thm: principal teorema o metodo que usa matrices binarias}
is the zero space, the inequalities implicated are linear rank inequalities.\footnote{One can use software such as Xitip to note that they most be Shannon
information inequalities.}
\end{cor}

Below is shown the class of $\left\lfloor \frac{n-1}{2}\right\rfloor -2$
inequalities that are true over finite sets of primes (i.e. sets of
the form $\left\{ p:p\mid t\right\} $), and another class of $\left\lfloor \frac{n-1}{2}\right\rfloor -2$
inequalities that are true over co-finite sets of primes (i.e. sets
of the form $\left\{ p:p\nmid t\right\} $).
\begin{example}
\label{ejemplo fuerte caracterisitica rango lineal para conjunto finito}Taking
$n\geq7$ and setting $t$ integer such that $2\leq t\leq\left\lfloor \frac{n-1}{2}\right\rfloor -1$
and $M\left(n,t\right)=n-t-2$, the following inequalities are produced
using as a guide square matrices $B_{M\left(n,t\right)}^{t}$with
column vectors of the form $B_{i}:=B_{e_{\left[M\left(n,t\right)\right]-i}}=c-e_{i}$,
$A_{i}:=A_{e_{i}}=e_{i}$, with $c=\underset{j\in\left[M\left(n,t\right)\right]}{\sum}e_{j},$
as described in figure \ref{fig figura}-left side. The rank of $B_{M\left(n,t\right)}^{t}$
is $M\left(n,t\right)$ when $\text{char}\mathbb{F}$ does not divide
$t$ and is $M\left(n,t\right)-1$ in other case. We remark that in
\cite{14} we used the case $M\left(n,t\right)=t+1$, so the columns
of the matrices were only of the form presented in figure \ref{fig figura}-right
side. Let $A_{1}$, $A_{2}$, $\ldots$, $A_{M\left(n,t\right)}$,
$B_{1}$, $B_{2}$, $\ldots$, $B_{t+1}$, $C$ be subspaces of a
finite-dimensional vector space $V$ over a scalar field $\mathbb{F}$.
We have:

(a) If $\text{char}\left(\mathbb{F}\right)$ divides $t$,
\[
\mathrm{H}\left(B_{\left[t+1\right]},A_{\left[M\left(n,t\right)\right]-\left[t+1\right]}\right)+\left(t+2\right)\left(M\left(n,t\right)-t-1\right)\mathrm{H}\left(C\right)\leq\left(M\left(n,t\right)-1\right)\mathrm{I}\left(A_{\left[M\left(n,t\right)\right]};C\right)
\]
\[
+\left(t+2\right)\stackrel[i=t+2]{M\left(n,t\right)}{\sum}\mathrm{H}\left(A_{i}\right)+\left[\left(t+2\right)\left(M\left(n,t\right)-t\right)-1\right]\left(\mathrm{H}\left(C\mid A_{\left[M\left(n,t\right)\right]}\right)+\stackrel[i=1]{M\left(n,t\right)}{\sum}\mathrm{I}\left(A_{\left[M\left(n,t\right)\right]-i};C\right)\right)
\]
\[
+\stackrel[i=1]{t+1}{\sum}\left(\mathrm{H}\left(B_{i}\mid A_{\left[M\left(n,t\right)\right]-i}\right)+\mathrm{H}\left(B_{i}\mid A_{i},C\right)+\mathrm{I}\left(A_{\left[i\right]};A_{\left[t+1\right]-\left[i\right]}\right)+\text{I}\left(A_{\left[i-1\right]};A_{i}\right)\right).
\]
(b) If $\text{char}\left(\mathbb{F}\right)$ does not divide $t$,
\[
\mathrm{H}\left(C\right)\leq\frac{1}{M\left(n,t\right)}\mathrm{H}\left(B_{\left[t+1\right]},A_{\left[M\left(n,t\right)\right]-\left[t+1\right]}\right)+\mathrm{H}\left(C\mid A_{\left[M\left(n,t\right)\right]}\right)+\stackrel[i=1]{M\left(n,t\right)}{\sum}\mathrm{I}\left(A_{\left[M\left(n,t\right)\right]-i};C\right)
\]
\[
+\stackrel[i=2]{t+1}{\sum}\mathrm{I}\left(A_{\left[i-1\right]};A_{i}\right)+\stackrel[i=1]{t+1}{\sum}\left(\mathrm{H}\left(C\mid A_{i},B_{i}\right)+\mathrm{H}\left(B_{i}\mid A_{\left[M\left(n,t\right)\right]-i}\right)+\mathrm{I}\left(A_{\left[i\right]};A_{\left[M\left(n,t\right)\right]-\left[i\right]}\right)\right).
\]
\end{example}

Corollary \ref{cor:corolario sobre eliminai=0000F3n de variable para que la desi sea rango lineal}
shows that each inequality, presented in example \ref{ejemplo fuerte caracterisitica rango lineal para conjunto finito},
can not be deduced from a higher order inequality by nullifying some
variables. In fact, using Corollary \ref{cor:corolario sobre cota de la dimension para que las desigualdades sean validas ind de caract},
we can say more about the class (a) of these inequalities.

For $m\in\mathbb{N}$ and $p$ prime, the function that counts all
the powers of $p$ less than or equal to $m$ is denoted by $\varphi\left(m,p\right)$.
In example \ref{ejemplo fuerte caracterisitica rango lineal para conjunto finito},
$\varphi\left(\left\lfloor \frac{n-1}{2}\right\rfloor -2,p\right)$
inequalities in $n$ variables, which are true over fields whose characteristic
is $p$, are produced. By Corollary \ref{cor:corolario sobre cota de la dimension para que las desigualdades sean validas ind de caract},
each of these inequalities holds over any characteristic when the
dimension of $V$ is at most $n-t-3$. Also, each inequality is determined
by $t$ and this number can run the powers of $p$ less than or equal
to $\left\lfloor \frac{n-1}{2}\right\rfloor -2$. This means that
each inequality is true in at least one vector space where the other
inequalities are not true. Therefore, any of these inequalities can
not be deduced from the other inequalities, much less if they are
combined with linear rank inequalities, without violating this property.
We have the next corollary.
\begin{cor}
For each $n\geq7$ and $p$ prime. There are at least $\varphi\left(\left\lfloor \frac{n-1}{2}\right\rfloor -2,p\right)$
independent inequalities in $n$ variables which are characteristic-dependent
linear rank inequalities that are true over fields whose characteristic
is $p$.
\end{cor}

\subsection{\label{subsec:secci=0000F3n de prueba del teorema principal}Proof
of Theorem \ref{thm: principal teorema o metodo que usa matrices binarias}:}

In a general way, we show how to build characteristic-dependent linear
rank inequalities from dependency relations in certain type of binary
matrices. We show this in three steps:

\textbf{A. Finding an equation.}

\textbf{B. Conditional characteristic-dependent linear rank inequalities.}

\textbf{C. Characteristic-dependent linear rank inequalities.}

First to all, we show how to abstract an equation as presented in
\cite[Lemma 3]{14}. Second, how to define ``conditional-linear rank
inequalities'' as presented in \cite[Lemma 5 and 6]{14}. Third,
the technique of upper bounds used in \cite[for a particular case]{4}
 and improved in \cite[for a family of binary matrices]{14} is applied.

\textbf{A. Finding an equation:} Let $\mathbb{F}^{n}=\left\langle e_{1}\right\rangle \oplus\cdots\oplus\left\langle e_{n}\right\rangle $
and $c=e_{1}+\cdots+e_{n}$. Let $B=\left(B^{i}\right)=\left(e_{S_{i}}\right)$
be a $n\times m$ binary matrix over $\mathbb{F}$, $m\leq n$. We
make the following correspondence between the columns of $B$ and
the canonical projection functions on $\mathbb{F}^{n}$: 
\[
e_{S_{i}}\dashleftarrow\dashrightarrow\pi_{S_{i}}\text{ \,\, where \ensuremath{S_{i}}=\ensuremath{\left\{  j:B_{\left(j,i\right)}=1\right\} } .}
\]

We suppose that $\text{rank}B=m$ if $\text{char}\mathbb{F}$ does
not divide $t$, and the $\text{rank}B=m-1$ if $\text{char}\mathbb{F}$
divides $t$, for $t\geq2$. Having account the previous correspondence,
we can define the following propositions whose proof is omitted:
\[
e_{S_{m}}=\stackrel[i=1]{m-1}{\sum}\alpha_{i}e_{S_{i}}\text{ if and only if }\pi_{S_{m}}=\stackrel[i=1]{m-1}{\sum}\alpha_{i}\pi_{S^{i}}
\]
\[
\left\{ e_{S_{i}}\right\} _{i=1}^{r}\text{is an independent set if and only if }\stackrel[i=1]{r}{\sum}\pi_{S_{i}}\left(\left\langle c\right\rangle \right)\text{ is a direct sum}
\]
We get an equation of the form: $\mathrm{H}\left(\pi_{S_{j}}\left(\left\langle c\right\rangle \right):j\in\left[m\right]\right)=\left\{ \begin{array}{c}
\left(m-1\right)\mathrm{H}\left(\left\langle c\right\rangle \right)\,\,\,\,\,\,\,\,\,\,\,\,\,\text{ if }\text{char}\left(\mathbb{F}\right)\mid t\\
m\mathrm{H}\left(\left\langle c\right\rangle \right)\text{ if }\text{char}\left(\mathbb{F}\right)\nmid t\text{.}
\end{array}\right.$

Previous argument can be easily generalized to vector subspaces $A_{e_{1}}$,
$...$, $A_{e_{n}}$, $C$ of a vector space $V$ over a field $\mathbb{F}$,
where $A_{e_{1}}$, $A_{e_{n}}$, $...$, $A_{e_{n}}$ are mutually
complementary in $V$ and $C$ is such that the sum of $\underset{i\neq k}{\bigoplus}A_{e_{i}}$
and $C$ is a direct sum for all $k$, such a collection of spaces
is called tuple that satisfies condition of complementary vector spaces.
Formally, we have:
\begin{claim}
\label{claim: ecuaci=0000F3n que separa rango}When $B$ exists, a
tuple that satisfies condition of complementary vector spaces holds
\[
\mathrm{H}\left(\pi_{S_{j}}\left(C\right):j\in\left[m\right]\right)=\left\{ \begin{array}{c}
\left(m-1\right)\mathrm{H}\left(C\right)\,\,\,\,\,\,\,\,\,\,\,\,\,\text{ if }\text{char}\left(\mathbb{F}\right)\mid t\\
m\mathrm{H}\left(C\right)\text{ if }\text{char}\left(\mathbb{F}\right)\nmid t\text{.}
\end{array}\right.
\]
\end{claim}

\begin{figure}[H]
\[
\begin{array}{c}
B_{1}\cdots B_{t+1}A_{t+2}\cdots A_{M\left(n,t\right)}\\
\left(\begin{array}{cccccc}
0 & \cdots & 1 & 0 & \cdots & 0\\
1 & \vdots & 1 & 0 & \vdots & 0\\
\vdots & \vdots & \vdots & \vdots & \vdots & 0\\
1 & \vdots & 1 & 0 & \vdots & 0\\
1 & \vdots & 0 & 0 & \vdots & \vdots\\
1 & \vdots & 1 & 1 & \vdots & 0\\
1 & \vdots & \vdots & 0 & \vdots & 0\\
1 & \vdots & 1 & \vdots & \vdots & 0\\
1 & \cdots & 1 & 0 & \cdots & 1
\end{array}\right)
\end{array}\,\,\,\,\,\,\,\,\begin{array}{c}
B_{1}\cdots B_{n+1}\\
\left(\begin{array}{ccc}
0 & \cdots & 1\\
1 & \vdots & 1\\
\vdots & \vdots & \vdots\\
1 & \vdots & 1\\
1 & \vdots & 0
\end{array}\right)
\end{array}
\]

\caption{\label{fig figura}Matrix $D_{M\left(n,t\right)}^{t}$ and matrix
$D_{M\left(2n+3,n\right)}^{n}$ used in \cite{14}.}
\end{figure}

\textbf{B. Conditional characteristic-dependent linear rank inequalities:}
In the previous step we noticed that dependence relations of $B$
can be expressed using projections of a suitable space $C$. In fact,
we can derive more properties as follows, from $e_{S_{i}}=\underset{j\in S_{i}}{\sum}e_{j}=c-\underset{j\notin S_{i}}{\sum}e_{j}$,
we derive
\[
\pi_{S_{i}}\left(C\right)=\underset{e_{j}\in S_{i}}{\bigoplus}A_{e_{j}}\cap\left(C\oplus\underset{e_{j}\notin S_{i}}{\bigoplus}A_{e_{j}}\right).
\]
This equality is easily proven. The following claim use this to find
inequalities that depend on the characteristic of $\mathbb{F}$, and
the involved spaces have some dependency relationships expressed by
$B$. We denote by $\mathcal{B}'=\left\{ e_{S_{i}}:1<\left|S_{i}\right|<n\right\} $;
$\mathcal{B}''=\left\{ e_{S_{i}}:\left|S_{i}\right|=1\right\} $;
$\mathcal{B}'''=\left\{ C\right\} $ if there exists $e_{S_{i}}$
in $B$ such that $\left|S_{i}\right|=n$, and $\mathcal{B}'''$ is
empty in other case.
\begin{claim}
\label{claim: desigualdades condicionales}For a tuple of vector subspaces
$\left(A_{e_{i}},B_{e_{S_{j}}},C:e_{S_{j}}\in\mathcal{B}',e_{i}\in\mathcal{B}''\right)$
such that $\left(A_{e_{1}},\ldots,A_{e_{n}},C\right)$ satisfies the
condition of complementary vector subspaces, consider the following
conditions:

(i) $A_{e_{i}}\leq A_{\left[e_{n}\right]-e_{i}}\oplus C$ for $i$
such that $e_{i}\in\mathcal{B}''$. ~~~~~~~~~(ii) $B_{e_{S_{i}}}\leq\underset{j\in S_{i}}{\bigoplus}A_{e_{j}}$
for $e_{S_{i}}\in\mathcal{B}'$.

(iii) $B_{e_{S_{i}}}\leq\underset{j\notin S_{i}}{\bigoplus}A_{e_{j}}\oplus C$
for $e_{S_{i}}\in\mathcal{B}'$.~~~~~~~~~~~~~~~~~~~~~~(iv)
$C\leq\underset{j\notin S_{i}}{\bigoplus}A_{e_{j}}+B_{e_{S_{i}}}$
for $e_{S_{i}}\in\mathcal{B}'$.

We have that

a. If condition (i), (ii) and (iii) hold over a fields whose characteristic
divides $t$, then 
\[
\text{H}\left(A_{e_{i}},B_{e_{S_{j}}},C:e_{S_{j}}\in\mathcal{B}',e_{i}\in\mathcal{B}'',C\in\mathcal{B}'''\right)\leq\left(m-1\right)\text{H}\left(C\right).
\]

b. If condition (ii) and (iv) hold over a fields whose characteristic
does not divide $t$, then 
\[
m\mathrm{H}\left(C\right)\leq\text{H}\left(A_{e_{i}},B_{e_{S_{j}}},C:e_{S_{j}}\in\mathcal{B}',e_{i}\in\mathcal{B}'',C\in\mathcal{B}'''\right).
\]
\end{claim}

\textbf{C. Characteristic-dependent linear rank inequalities:} We
find vector subspaces that satisfy conditions of previous claim. Let
$\left(A_{e_{i}},B_{e_{S_{j}}},C:e_{S_{j}}\in\mathcal{B}',e_{i}\in\mathcal{B}''\right)$
a tuple of arbitrary vector subspaces of a finite dimensional vector
space $V$ over a finite field $\mathbb{F}$.

From $A_{e_{1}}$, $\ldots$, $A_{e_{n}}$, and $C$, we obtain a
tuple that satisfies condition of complementary spaces $\left(A_{e_{1}}',\ldots,A_{e_{n}}',\bar{C}\right)$
as obtained in \cite{14} which holds:
\begin{equation}
\text{codim}_{A_{e_{k}}}\left(A_{e_{k}}'\right)=\text{I}\left(A_{\left[e_{k-1}\right]};A_{e_{k}}\right)\text{, for all \ensuremath{k},}\label{ecuaciones de codimension de los Ai primas-1}
\end{equation}
\begin{equation}
\text{codim}_{C}\left(\bar{C}\right)\leq\nabla\left(C\right)\overset{\text{def.}}{=}\text{H}\left(C\mid A_{\left[e_{n}\right]}\right)+\underset{e_{i}\in\left[e_{n}\right]}{\sum}\text{I}\left(A_{\left[e_{n}\right]-e_{i}};C\right)\text{.}\label{eq: codimension de C (raya horizontal arriba)-1}
\end{equation}
Additionally, for $T\subseteq\left[e_{n}\right]$, we can take some
elements $e_{k_{1}},\ldots,e_{k_{l}}$ with $k_{1}\leq k_{2}\leq\cdots\leq k_{l}$
such that it is possible to built a partition in intervals $\left[e_{k_{i}},e_{k_{j}}\right]$,
with maximum length, of $T$. So,
\begin{equation}
\text{codim}_{\underset{e\in T}{\sum}A_{e}}\underset{e\in T}{\bigoplus}A_{e}'\leq\nabla\left(A_{e}:e\in T\right)\overset{\text{def.}}{=}\text{I}\left(A_{\left[e_{1},e_{k_{1}}\right)};A_{\left[e_{k_{1}},e_{k_{2}}\right]}\right)+\cdots+\text{I}\left(A_{\left[e_{1},e_{k_{1}-1}\right)};A_{\left[e_{k_{1-1}},e_{k_{l}}\right]}\right).\label{eq:desigualdad de triangulo}
\end{equation}
Before continuing, we need the following three statements:
\begin{claim}
\label{claim: tupla que cumle (i)}Tuple $\left(\bar{A}_{e_{1}},\ldots,\bar{A}_{e_{n}}\bar{C}\right)$
defined by
\[
\bar{A}_{e_{k}}:=A_{e_{k}}'\cap\left(\bar{C}+\underset{e_{i}\notin D''}{\bigoplus}A_{e_{i}}'+\underset{e_{i}\in D'',i<k}{\bigoplus}\bar{A}_{e_{i}}+\underset{e_{i}\in D'',i>k}{\bigoplus}A_{e_{i}}'\right),\text{ for \ensuremath{e_{k}\in\mathcal{B}''}}
\]
\[
\bar{A}_{e_{k}}:=A_{e_{k}}',\text{ for \ensuremath{e_{k}\notin\mathcal{B}''}}
\]
satisfies condition of complementary spaces, condition (i), $\text{H}\left(\bar{A}_{e_{k}}\right)=\text{H}\left(\bar{C}\right)\leq1$
and
\begin{equation}
\text{codim}_{A_{e_{k}}}\left(\bar{A}_{e_{k}}\right)\leq\text{H}\left(A_{e_{k}}\right)-\text{H}\left(C\right)+\nabla\left(C\right),\text{ for \ensuremath{e_{k}\in\mathcal{B}''}}.\label{eq: desigualdad de codimension de A raya y A para B''}
\end{equation}
\end{claim}

\begin{proof}
We obviously have that $\bar{A}_{e_{k}}\leq\bar{C}+\underset{e_{i}\notin\mathcal{B}''}{\bigoplus}A_{e_{i}}'+\underset{e_{i}\in\mathcal{B}'',i<k}{\bigoplus}\bar{A}_{e_{i}}+\underset{e_{i}\in\mathcal{B}'',i>k}{\bigoplus}A_{e_{i}}'$.
Now, for $k$, we show that 
\begin{equation}
\bar{C}\leq\underset{e_{i}\notin\mathcal{B}''}{\bigoplus}A_{e_{i}}'+\underset{e_{i}\in\mathcal{B}'',i\leq k}{\bigoplus}\bar{A}_{e_{i}}+\underset{e_{i}\in\mathcal{B}'',i>k}{\bigoplus}A_{e_{i}}'\label{eq: contenencia de C raya a cada k}
\end{equation}
In effect, we show case $k=l=\min\left\{ i:e_{i}\in\mathcal{B}''\right\} $,
the general case is by induction. We note that $\bar{A}_{e_{l}}=O$
if and only if $\bar{C}=O$, so this case is trivial. Otherwise, there
exists $a_{e_{l}}\neq O$ in $\bar{A}_{e_{l}}$ and $c\neq O$ in
$\bar{C}$ such that $a_{e_{l}}=c+\underset{i\neq l}{\sum}a_{e_{i}}$,
then $c\in\underset{e_{i}\notin\mathcal{B}''}{\bigoplus}A_{e_{i}}'+\bar{A}_{e_{l}}+\underset{e_{i}\in\mathcal{B}'',i>l}{\bigoplus}A_{e_{i}}'$.
The affirmation is obtained noting that $c$ generates $\bar{C}$
\cite[Remark 4]{14}. Taking $k=\max\left\{ i:e_{i}\in\mathcal{B}''\right\} $,
we obtain that $\bar{C}\leq\bigoplus\bar{A}_{e_{i}}$, so condition
of complementary spaces is satisfied. Also, we have the equation:
\[
\text{H}\left(\bar{A}_{e_{k}}\right)=\text{I}\left(A'_{e_{k}};\bar{C},\underset{e_{i}\notin\mathcal{B}''}{\bigoplus}A_{e_{i}}',\underset{e_{i}\in\mathcal{B}'',i<k}{\bigoplus}\bar{A}_{e_{i}},\underset{e_{i}\in\mathcal{B}'',i>k}{\bigoplus}A_{e_{i}}'\right)
\]
\[
=\text{H}\left(A'_{e_{k}}\right)-\text{H}\left(\underset{e_{i}\notin\mathcal{B}''}{\bigoplus}A_{e_{i}}',\underset{e_{i}\in\mathcal{B}'',i<k}{\bigoplus}\bar{A}_{e_{i}},\underset{e_{i}\in\mathcal{B}'',i\geq k}{\bigoplus}A_{e_{i}}'\right)+\text{H}\left(\bar{C},\underset{e_{i}\notin\mathcal{B}''}{\bigoplus}A_{e_{i}}',\underset{e_{i}\in\mathcal{B}'',i<k}{\bigoplus}\bar{A}_{e_{i}},\underset{e_{i}\in\mathcal{B}'',i>k}{\bigoplus}A_{e_{i}}'\right)
\]
\[
\,\,\,\,\,\,\,\,\,\,\,\,\,\,\,\,\,\,\,\,\,\,\,\,\,\,\,\,\,\,\,\,\,\,\,\,\,\,\,\,\,\,\,\,\,\,\,\,\,\,\,\,\,\,\,\,\,\,\,\,\,\,\,\,\,\,\,\,\,\,\,\,\,\,\,\,\,\,\,\,\,\,\,\,\,\,\,\,\,\,\,\,\,\,\,\,\,\,\,\,\,\,\,\,\,\,\,\,\,\,\,\,\,\,\,\,\,\,\,\,\,\,\,\,\,\,\,\,\,\,\,\,\,\,\,\,\,\,\,\,\,\,\,\,\,\,\,\,\,\,\,\,\,\,\,\,\,\,\,\,\,\,\,\,\,\,\,\,\,\,\,\,\,\,\,\,\,\,\:\,\text{ [from (\ref{eq: contenencia de C raya a cada k})]}
\]
\[
=\text{H}\left(\bar{C}\right)\,\,\,\,\,\,\,\,\,\,\,\,\,\,\,\,\,\,\,\,\,\,\,\,\,\,\,\,\,\,\,\text{ [using condition of complementary spaces]}
\]
which also implies the desired upper bound on $\text{codim}_{A_{e_{k}}}\left(\bar{A}_{e_{k}}\right)$.
Now, condition (i) is straightforward because $\bar{C}$ and each
$\bar{A}_{e_{k}}$, $\ensuremath{e_{k}\in\mathcal{B}''}$, have the
same dimension.
\end{proof}
\begin{claim}
\label{claim: tupla que cumple (i)  (ii) (iii)}For $e_{S_{k}}\in\mathcal{B}'$,
we define $\bar{B}_{e_{S_{k}}}:=B_{e_{S_{k}}}\cap\left(\underset{e_{i}\in S_{k}}{\bigoplus}\bar{A}_{e_{i}}\right)\cap\left(\underset{e_{i}\notin S_{k}}{\bigoplus}\bar{A}_{e_{i}}\oplus\bar{C}\right)$.
We have that tuple $\left(\bar{A}_{e_{1}},\ldots\bar{A}_{e},\bar{B}_{e_{S_{k}}},\bar{C}:e_{S_{k}}\in\mathcal{B}'\right)$
satisfies condition (i), condition (ii), condition (iii) and

\[
\text{codim}_{B_{e_{S_{k}}}}\bar{B}_{e_{S_{k}}}\leq\text{H}\left(B_{e_{S_{k}}}\mid A_{e_{i}}:i\in S_{k}\right)+\text{H}\left(B_{e_{S_{k}}}\mid A_{e_{i}},C:i\notin S_{k}\right)+\underset{e_{i}\in\mathcal{B}''}{\sum}\text{H}\left(A_{e_{i}}\right)
\]
\[
+\nabla\left(A_{e_{i}}:i\in S_{k},e_{i}\notin\mathcal{B}''\right)+\nabla\left(A_{e_{i}}:i\notin S_{k},e_{i}\notin\mathcal{B}''\right)+\left(\left|\mathcal{B}''\right|+1\right)\nabla\left(C\right)-\left|\mathcal{B}''\right|\text{H}\left(C\right)
\]
\end{claim}

\begin{proof}
In effect,
\[
\text{codim}_{B_{e_{S_{k}}}}\bar{B}_{e_{S_{k}}}\leq\text{codim}_{B_{e_{S_{k}}}}\left(\left(\underset{i\in S_{k}}{\bigoplus}\bar{A}_{e_{i}}\right)\cap B_{e_{S_{k}}}\right)+\text{codim}_{B_{e_{S_{k}}}}\left(\left[\underset{i\notin S_{k}}{\bigoplus}\bar{A}_{e_{i}}\oplus\bar{C}\right]\cap B_{e_{S_{k}}}\right)
\]
\[
=\text{codim}_{B_{e_{S_{k}}}}\left(\left(\underset{i\in S_{k}}{\bigoplus}A_{e_{i}}\right)\cap B_{e_{S_{k}}}\right)+\text{codim}_{B_{e_{S_{k}}}}\left(\left[\underset{i\notin S_{k}}{\bigoplus}A_{e_{i}}\oplus\bar{C}\right]\cap B_{e_{S_{k}}}\right)
\]
\[
+\text{codim}_{\left(\underset{i\in S_{k}}{\sum}A_{e_{i}}\right)\cap B_{e_{S_{k}}}}\left(\left(\underset{i\in S_{k}}{\bigoplus}\bar{A}_{e_{i}}\right)\cap B_{e_{S_{k}}}\right)+\text{codim}_{\left[\underset{i\notin S_{k}}{\sum}A_{e_{i}}\oplus\bar{C}\right]\cap B_{e_{S_{k}}}}\left(\left[\underset{i\notin S_{k}}{\bigoplus}\bar{A}_{e_{i}}\oplus\bar{C}\right]\cap B_{e_{S_{k}}}\right)
\]
\[
\leq\text{H}\left(B_{e_{S_{k}}}\mid A_{e_{i}}:i\in S_{k}\right)+\text{H}\left(B_{e_{S_{k}}}\mid C,A_{e_{i}}:i\notin S_{k}\right)+
\]
\[
+\text{codim}_{\underset{i\in S_{k},e_{i}\in\mathcal{B}''}{\sum}A_{e_{i}}}\underset{i\in S_{k},e_{i}\in\mathcal{B}''}{\bigoplus}\bar{A}_{e_{i}}+\text{codim}_{\underset{i\in S_{k},e_{i}\notin\mathcal{B}''}{\sum}A_{e_{i}}}\underset{i\in S_{k},e_{i}\notin\mathcal{B}''}{\bigoplus}A_{e_{i}}+
\]
\[
\text{codim}_{\underset{i\notin S_{k},e_{i}\in\mathcal{B}''}{\sum}A_{e_{i}}\oplus\bar{C}}\underset{i\notin S_{k},e_{i}\in\mathcal{B}''}{\bigoplus}\bar{A}_{e_{i}}\oplus\bar{C}+\text{codim}_{\underset{i\notin S_{k},e_{i}\notin\mathcal{B}''}{\sum}A_{e_{i}}}\underset{i\notin S_{k},e_{i}\notin\mathcal{B}''}{\bigoplus}A_{e_{i}}+\text{codim}_{\bar{C}}\bar{C}
\]
\[
\leq\text{H}\left(B_{e_{S_{k}}}\mid A_{e_{i}}:i\in S_{k}\right)+\text{H}\left(B_{e_{S_{k}}}\mid A_{e_{i}},C:i\notin S_{k}\right)+\underset{i\in S_{k},e_{i}\in\mathcal{B}''}{\sum}\text{H}\left(A_{e_{i}}\right)
\]
\[
+\left|\left\{ e_{i}\in\mathcal{B}'':i\in S_{k}\right\} \right|\left(\nabla\left(C\right)-\text{H}\left(C\right)\right)+\nabla\left(A_{e_{i}}:i\in S_{k},e_{i}\notin\mathcal{B}''\right)+\underset{i\notin S_{k},e_{i}\in\mathcal{B}''}{\sum}\text{H}\left(A_{e_{i}}\right)
\]
\[
+\left|\left\{ e_{i}\in\mathcal{B}'':i\notin S_{k}\right\} \right|\left(\nabla\left(C\right)-\text{H}\left(C\right)\right)+\nabla\left(A_{e_{i}}:i\notin S_{k},e_{i}\notin\mathcal{B}''\right)+\nabla\left(C\right)
\]
\[
\,\,\,\,\,\,\,\,\,\,\,\,\,\,\,\,\,\,\,\,\,\,\,\,\text{ [from Lemmas 1 and 2 of [8] , inequalities (\ref{eq:desigualdad de triangulo}) and (\ref{eq: desigualdad de codimension de A raya y A para B''})].}
\]
\[
=\text{H}\left(B_{e_{S_{k}}}\mid A_{e_{i}}:i\in S_{k}\right)+\text{H}\left(B_{e_{S_{k}}}\mid A_{e_{i}},C:i\notin S_{k}\right)+\underset{e_{i}\in\mathcal{B}''}{\sum}\text{H}\left(A_{e_{i}}\right)
\]
\[
+\left(\left|\mathcal{B}''\right|+1\right)\nabla\left(C\right)-\left|\mathcal{B}''\right|\text{H}\left(C\right)+\nabla\left(A_{e_{i}}:i\in S_{k},e_{i}\notin\mathcal{B}''\right)+\nabla\left(A_{e_{i}}:i\notin S_{k},e_{i}\notin\mathcal{B}''\right)
\]
\end{proof}
\begin{claim}
\label{claim: tupla que cumple (ii) (iv)}For $e_{S_{k}}\in\mathcal{B}'$,
we define $\hat{B}_{e_{S_{k}}}:=B_{e_{S_{k}}}\cap\underset{j\in S_{k}}{\bigoplus}A_{e_{j}}'$
and $\hat{C}:=\bar{C}\underset{e_{S_{k}}}{\bigcap}\left(\underset{j\notin S_{i}}{\bigoplus}A_{e_{j}}'+\hat{B}_{e_{S_{k}}}\right)$.
We have that $\left(A_{e_{i}}',\hat{B}_{e_{S_{j}}},\hat{C}:e_{e_{i}}\in\mathcal{B}'',e_{S_{j}}\in\mathcal{B}'\right)$
satisfies condition of complementary spaces, conditions (ii), (iv)
and
\begin{equation}
\text{codim}_{B_{e_{S_{k}}}}\hat{B}_{e_{S_{k}}}\leq\text{H}\left(B_{e_{S_{k}}}\mid A_{e_{i}}:i\in S_{k}\right)+\nabla\left(A_{e_{i}}:i\in S_{k}\right),\label{eq: desigualdad de codimension de B gorrito y B}
\end{equation}
\[
\text{codim}_{C}\hat{C}\leq\nabla\left(C\right)+\underset{e_{S_{k}}\in\mathcal{B}'}{\sum}\left[\text{H}\left(C\mid A_{e_{i}},B_{S_{k}}:i\notin S_{k}\right)+\text{H}\left(B_{e_{S_{k}}}\mid A_{e_{i}}:i\in S_{k}\right)\right]
\]
\begin{equation}
+\underset{e_{S_{k}}\in\mathcal{B}'}{\sum}\left[+\nabla\left(A_{e_{i}}:i\notin S_{k}\right)+\nabla\left(A_{e_{i}}:i\in S_{k}\right)\right].\label{eq:desigualdad de codimension de C gorrito a C}
\end{equation}
\end{claim}

\begin{proof}
We only show last inequality:
\[
\text{codim}_{C}\hat{C}\leq\text{codim}_{C}\bar{C}+\underset{e_{S_{k}}\in\mathcal{B}'}{\sum}\text{codim}_{C}\left(C\cap\left[\underset{i\notin S_{k}}{\bigoplus}A_{e_{i}}'+\hat{B}_{e_{S_{k}}}\right]\right)
\]
\[
=\text{codim}_{C}\bar{C}+\underset{e_{S_{k}}\in\mathcal{B}'}{\sum}\text{codim}_{C}\left(C\cap\left[\underset{i\notin S_{k}}{\bigoplus}A_{e_{i}}+B_{e_{S_{k}}}\right]\right)
\]
\[
+\underset{e_{S_{k}}\in\mathcal{B}'}{\sum}\text{codim}_{C\cap\left[\underset{i\notin S_{k}}{\sum}A_{e_{i}}+B_{e_{S_{k}}}\right]}\left(C\cap\left[\underset{i\notin S_{k}}{\bigoplus}A_{e_{i}}'+\hat{B}_{e_{S_{k}}}\right]\right)
\]
\[
\leq\text{codim}_{C}\bar{C}+\underset{e_{S_{k}}\in\mathcal{B}'}{\sum}\text{codim}_{C}\left(C\cap\left[\underset{i\notin S_{k}}{\bigoplus}A_{e_{i}}+B_{e_{S_{k}}}\right]\right)+\underset{e_{S_{k}}\in\mathcal{B}'}{\sum}\text{codim}_{\underset{i\notin S_{k}}{\sum}A_{e_{i}}}\underset{i\notin S_{k}}{\bigoplus}A_{e_{i}}'
\]
\[
+\underset{e_{S_{k}}\in\mathcal{B}'}{\sum}\text{codim}_{B_{e_{S_{k}}}}\hat{B}_{e_{S_{k}}}\text{ [from Lemmas 1 and 2 of [8] and inequality (\ref{eq: desigualdad de codimension de B gorrito y B})]}
\]
\[
\leq\nabla\left(C\right)+\underset{e_{S_{k}}\in\mathcal{B}'}{\sum}\left[\text{H}\left(C\mid A_{e_{i}},B_{S_{k}}:i\notin S_{k}\right)+\text{H}\left(B_{e_{S_{k}}}\mid A_{e_{i}}:i\in S_{k}\right)\right]
\]
\[
+\underset{e_{S_{k}}\in\mathcal{B}'}{\sum}\left[+\nabla\left(A_{e_{i}}:i\notin S_{k}\right)+\nabla\left(A_{e_{i}}:i\in S_{k}\right)\right]\text{ [from (\ref{eq:desigualdad de triangulo})}\text{]}
\]
\end{proof}
We can finally build the inequalities of our theorem:

By Claims \ref{claim: tupla que cumle (i)} and \ref{claim: tupla que cumple (i)  (ii) (iii)},
tuple $\left(\bar{A}_{e_{i}},\bar{B}_{e_{S_{j}}},\bar{C}:e_{i}\in\left[e_{n}\right],e_{S_{j}}\in\mathcal{B}'\right)$
satisfies hypothesis of the Claim \ref{claim: desigualdades condicionales}
with conditions (i), (ii) and (iii) in a finite field $\mathbb{F}$
whose field characteristic divides $t$, we get
\begin{equation}
\text{H}\left(\bar{B}_{e_{S_{i}}},\bar{A}_{e_{j}},\bar{C}:e_{S_{i}}\in\mathcal{B}',e_{j}\in\mathcal{B}'',\bar{C}\in\mathcal{B}'''\right)\leq\left(m-1\right)\text{H}\left(\bar{C}\right)\text{.}\label{eq:-17}
\end{equation}
On the other hand,
\begin{equation}
\text{H}\left(\bar{C}\right)\leq\text{I}\left(A_{\left[e_{n}\right]};C\right)\text{\,\,\,\,[from \ensuremath{\bar{C}\leq C}]},\label{eq:-18}
\end{equation}
\[
\text{codim}_{\underset{e_{S_{k}}\in\mathcal{B}'}{\sum}B_{S_{i}}}\underset{e_{S_{k}}\in\mathcal{B}'}{\sum}\bar{B}_{S_{1}}\leq\underset{e_{S_{k}}\in\mathcal{B}'}{\sum}\text{codim}_{B_{e_{S_{i}}}}\bar{B}_{e_{S_{i}}}\text{\,\,\,\,[from Lemmas 1 of [8]].}
\]
Therefore,
\[
\text{codim}_{\underset{e_{S_{k}}\in\mathcal{B}'}{\sum}B_{S_{i}}+\underset{e_{i}\in\mathcal{B}''}{\sum}A_{e_{i}}}\left(\underset{e_{S_{k}}\in\mathcal{B}'}{\sum}\bar{B}_{S_{i}}+\underset{e_{i}\in\mathcal{B}''}{\sum}\bar{A}_{e_{i}}\right)\leq\underset{e_{S_{k}}\in\mathcal{B}'}{\sum}\text{H}\left(B_{e_{S_{k}}}\mid A_{e_{i}}:i\in S_{k}\right)
\]
\[
\underset{e_{S_{k}}\in\mathcal{B}'}{\sum}\text{H}\left(B_{e_{S_{k}}}\mid A_{e_{i}},C:i\notin S_{k}\right)+\left(\left|\mathcal{B}'\right|+1\right)\underset{e_{i}\in\mathcal{B}''}{\sum}\text{H}\left(A_{e_{i}}\right)
\]
\[
+\left(\left|\mathcal{B}''\right|\left|\mathcal{B}'\right|+\left|\mathcal{B}''\right|+\left|\mathcal{B}'\right|\right)\nabla\left(C\right)-\left(\left|\mathcal{B}''\right|\left|\mathcal{B}'\right|+\left|\mathcal{B}''\right|\right)\text{H}\left(C\right)
\]
\[
+\underset{e_{S_{k}}\in\mathcal{B}'}{\sum}\left[\nabla\left(A_{e_{i}}:i\in S_{k},e_{i}\notin\mathcal{B}''\right)+\nabla\left(A_{e_{i}}:i\notin S_{k},e_{i}\notin\mathcal{B}''\right)\right].
\]
From (\ref{eq:-17}) , (\ref{eq:-18}), (\ref{eq: codimension de C (raya horizontal arriba)-1})
and last inequality, we can obtain the desired characteristic-dependent
linear rank inequality over fields whose characteristic divides $t$.

By Claims \ref{claim: tupla que cumple (ii) (iv)}, tuple $\left(A_{e_{i}}',\hat{B}_{e_{S_{j}}},\hat{C}:e_{i}\in\left[e_{n}\right],e_{S_{j}}\in\mathcal{B}'\right)$
satisfies hypothesis of the Claim \ref{claim: desigualdades condicionales}
with conditions (ii) and (iv) in a finite field $\mathbb{F}$ whose
field characteristic does not divide $t$, we get
\begin{equation}
m\text{H}\left(\hat{C}\right)\leq\text{H}\left(A_{e_{i}}',\hat{B}_{e_{S_{j}}},\hat{C}:e_{S_{j}}\in\mathcal{B}',e_{i}\in\mathcal{B}'',\hat{C}\in\mathcal{B}'''\right).\label{eq:Caso distinto de 3 desigualdad condicional basica-1}
\end{equation}
On the other hand,
\begin{equation}
\text{H}\left(A_{e_{i}}',\hat{B}_{e_{S_{j}}},\hat{C}:e_{S_{j}}\in\mathcal{B}',e_{i}\in\mathcal{B}'',\hat{C}\in\mathcal{B}'''\right)\leq\text{H}\left(A_{e_{i}},B_{e_{S_{j}}},C:e_{S_{j}}\in\mathcal{B}',e_{i}\in\mathcal{B}'',C\in\mathcal{B}'''\right).\label{eq:-20}
\end{equation}
From (\ref{eq:Caso distinto de 3 desigualdad condicional basica-1})
, (\ref{eq:desigualdad de codimension de C gorrito a C}) and last
inequality, we can obtain the desired characteristic-dependent linear
rank inequality over fields whose characteristic does not divide $t$.
\begin{rem}
In case that the dimension of $V$ is at most $n-1$, there exists
some $A_{e_{i}}'=O$ in the demonstration above. Therefore, the equation
given by the matrix used as a guide is trivial. This implies Corollary
\ref{cor:corolario sobre cota de la dimension para que las desigualdades sean validas ind de caract}.
Corollary \ref{cor:corolario sobre eliminai=0000F3n de variable para que la desi sea rango lineal}
is obtained in a similar way.
\end{rem}

\section*{Acknowledgments}

The first author thanks the support provided by COLCIENCIAS.

\end{document}